\documentclass[a4,a4wide]{article}

\usepackage{amsfonts}
\usepackage{amsmath}
\usepackage{amssymb}
\usepackage{courier}
\usepackage{datetime}
\usepackage{epsfig}
\usepackage{graphicx}
\usepackage{helvet}
\usepackage{paralist}
\usepackage{sidecap}
\usepackage{stmaryrd}
\usepackage{tikz} 
\usepackage{times}
\usepackage{url}
\usepackage{verbatim}
\usepackage{wrapfig}

{\newenvironment{proof}
  {\setlength{\parindent}{0pt}{\bf Proof:}}
    {\rule{1.4mm}{1.9mm}}

\newcommand{\ignore}[1]{} 

\newcommand{\cA}{{\cal A}}
\newcommand{\cB}{{\cal B}}

\newcommand{\modop}[1]{{\llbracket #1 \rrbracket}}

\DeclareMathOperator*{\MS}{MS}

\newcommand{\partitle}[1]{{\noindent \bf #1}}

\makeatletter
\providecommand*{\cupdot}{%
  \mathbin{%
    \mathpalette\@cupdot{}%
  }%
}
\newcommand*{\@cupdot}[2]{%
  \ooalign{%
    $\m@th#1\cup$\cr
    \hidewidth$\m@th#1\cdot$\hidewidth
  }%
}
\makeatother

\usetikzlibrary{calc,shadows,arrows,decorations.pathreplacing}
\pgfdeclarelayer{background}
\pgfdeclarelayer{foreground}
\pgfsetlayers{background,main,foreground}
 
\tikzstyle{materia}=[draw, fill=blue!20, text width=6.0em, text centered,
  minimum height=1.5em,drop shadow]
\tikzstyle{module} = [materia, text width=8em, minimum width=10em,
  minimum height=3em, rounded corners, drop shadow]
\tikzstyle{texto} = [above, text width=6em, text centered]
\tikzstyle{linepart} = [draw, thick, color=black!50, -latex', dashed]
\tikzstyle{line} = [draw, thick, color=black!50, -latex']
\tikzstyle{ur}=[draw, text centered, minimum height=0.01em]

\newcommand{\primitivemodule}[2]{node (mod#1) [module] {#2}}

\frenchspacing
\pagestyle{empty}

\usepackage{aaai/aaai}

\ignore{

\usepackage{times}
\usepackage{helvet}
\usepackage{amsfonts}
\usepackage{amsmath}
\usepackage{amssymb}
\usepackage{amsthm}
\usepackage{courier}
\usepackage{paralist}
\usepackage{newlfont}

\usepackage{graphicx}
\usepackage{epsfig}
\usepackage{wrapfig}

\def\cA{\ensuremath{{\cal{A}}}}
\def\cB{\ensuremath{{\cal{B}}}}

\newcommand{\ignore}[1]{}

\newcommand{\partitle}[1]{{\noindent \bf #1 -- }}

}

\newtheorem{corollary}{Corollary}
\newtheorem{definition}{Definition}
\newtheorem{example}{Example}

\newtheorem{proposition}{Proposition}
\newtheorem{remark}{Remark}
\newtheorem{theorem}{Theorem}

\newcounter{inlineequation}
\setcounter{inlineequation}{0}

\title{Three Semantics for Modular Systems }
\author{Shahab Tasharrofi and Eugenia Ternovska\\Simon Fraser
University, email: \{ter, sta44\}@cs.sfu.ca}

\begin{document}
\nocopyright
\maketitle

\begin{abstract}

In this paper, we further develop the framework of Modular Systems that lays model-theoretic foundations for combining different declarative languages, agents and solvers. 
 We introduce a  multi-language logic of modular systems. We define two novel semantics, a structural operational semantics, and an inference-based semantics. We prove the new semantics are equivalent to the original model-theoretic semantics and describe future research directions.

\end{abstract}

\section{Introduction}\label{sec:introduction}

Modular Systems (MS)  \cite{TT:FROCOS:2011-long} is a language-independent formalism 
representing and  solving  complex problems specified declaratively.
There are several motivations for introducing the MS formalism: 
\begin{compactitem}
\item  the need to be able to
split a large problem into subproblems, and to use the most suitable formalism
for each part,
\item the need to model distributed combinations of programs, knowledge
bases, languages, agents, etc.,
\item the need to model collaborative solving of complex tasks, 
such as in satisfiability-based solvers.
\end{compactitem}
 The MS formalism gave a unifying view, through a semantic approach, to formal and declarative modelling of modular systems.
In that initial work, individual modules were considered from both model-theoretic and
operational view. Under the model-theoretic view, a module is a set (or
class) of structures, and under the operational view it is an operator, mapping
a subset of the vocabulary to another subset. An abstract algebra on modules
was given. It is similar to Codd's relational algebra and allows one to combine modules on abstract model-theoretic
level, independently from what languages are used for describing them. 
An important operation in the algebra is the loop (or feedback)
operation, since iteration underlies many solving methods. We showed 
that the power of the loop operator is such that the combined modular system
can capture all of the complexity class NP even when each module is
deterministic and polytime. Moreover, in general, adding loops gives a jump in
the polynomial time hierarchy, one step from the highest complexity of the
components. It is also shown that each module can be viewed as an operator, and
when each module is (anti-) monotone, the number of the potential solutions can
be reduced by using ideas from the logic programming community.

Inspired by practical combined solvers,
the authors of \cite{TWT:WLP:2011,TWT:WLP-INAP:2012}  introduced an algorithm to solve model expansion tasks for modular systems. The evolution processes of different modules are jointly considered.
The algorithm incrementally constructs structures for the expanded vocabulary by communicating with oracles associated 
with each module, who provide additional information in the form of reasons 
and advice to navigate the search.
It was shown that the algorithm closely corresponds to what is done in practice in  different areas such as Satisfiability Modulo Theories (SMT),
Integer Linear Programming (ILP), Answer Set Programming (ASP).

\ignore{
Here, we develop 
the framework of modular systems that Hence, 
users are not restricted to an individual language anymore and can now define each 
sub-part of their problems in a different language while still being able to combine these 
languages and to solve the main problem.

ntroduce and investigate a modular approach towards specifying and solving complex tasks so that different and possibly 
hetero- geneous parts can work together.

\subsubsection{motivation for addressing modularity}

The authors of .......... \cite{TT:FRACOS:2011}    proposed a model-theoretic approach to multi-language modular systems 
I
In this work, we continue this line of research originated by \cite{J,TT:2011} to extend the MX framework to modular systems.

}
\subsubsection{Background: Model Expansion} In  \cite{MT05}, the authors  formalize combinatorial search problems  
 as the  task of
{\em model expansion (MX)}, the logical task of expanding a given (mathematical)
structure with new relations. 
Formally, the user axiomatizes the problem in
some logic $\cal L$. This axiomatization relates an instance of the problem (a
{\em finite structure}, i.e., a universe together with some relations and functions),
and its solutions (certain {\em expansions} of that structure with new relations or
functions). Logic $\cal L$ corresponds to a specification/modelling language. It
could be an extension of first-order logic such as FO(ID), or an ASP language,
or a modelling language from the CP community such as ESSENCE \cite{ESSENCE}.
The MX framework was later extended to infinite structures to formalise built-in arithmetic 
in specification languages \cite{TM:IJCAI:2009,TT:arithmetic:LPAR-17}.

Recall that a vocabulary is a set of non-logical (predicate and function)
symbols. An interpretation for a vocabulary is provided by a {\em structure},
which consists of a set, called the domain or universe and denoted by $dom(.)$,
together with a collection of relations and (total) functions over the universe.
A structure can be viewed as an {\em assignment} to the elements of
the vocabulary. An expansion of a structure $\cal A$ is a structure $\cal B$
with the same universe, and which has all the relations and functions of
$\cal A$, plus some additional relations or functions.

Formally, the task of model expansion for an arbitrary logic $\cal L$  is:
Given an $\cal L$-formula $\phi$ with vocabulary
  $\sigma \cup \varepsilon$ and a structure $\cal A$ for $\sigma$
find an expansion of $\cal A$, to $\sigma \cup \varepsilon$,
   that satisfies $\phi$.
Thus, we expand the structure $\cal A$ with relations and functions to interpret
$\varepsilon$, obtaining a model $\cB$ of $\phi$. 
We call $\sigma$, the
vocabulary of $\cal A$, the {\em instance} vocabulary, and
$\varepsilon := vocab(\phi)\setminus \sigma$ the {\em expansion} vocabulary\footnote{By ``$:=$'' we mean ``is by definition'' or ``denotes''. By $vocab(\phi)$ we understand the vocabulary of $\phi$.}.  If $\sigma=\emptyset$, we talk about {\em model generation}, a particular type of model expansion that 
is often studied. 

Given a specification, we can talk about a set of
$\sigma \cup \varepsilon$-structures which satisfy the specification.
Alternatively, we can simply talk about a given set of
$\sigma \cup \varepsilon$-structures as an MX-task, without mentioning a
particular specification the structures satisfy. These sets of structures will 
be called {\em modules} later in the paper.
This abstract view makes our
study of modularity language-independent.

%

\begin{example}\label{ex:colouring}
The following logic program $\phi$ constitutes an MX specification for Graph
3-colouring:

\begin{small}
$$
\begin{array}{c}
1\{R(x),B(x),G(x)\}1 \leftarrow V(x). \\
\bot \leftarrow R(x), R(y), E(x,y). \\
\bot \leftarrow B(x), B(y), E(x,y). \\
\bot \leftarrow G(x), G(y), E(x,y). \\
\end{array}
$$
\end{small}

An instance is a structure for vocabulary $\sigma = \{E\}$, i.e., a graph
${\cal A} = {\cal G} = ( V; E)$. The task is to find an interpretation for the
symbols of the expansion vocabulary $\varepsilon =\{ R, B, G\}$ such that the
expansion of $\cA$ with these is a model of $\phi$:
\begin{small}

$$
\underbrace{\overbrace{(V; E^{\cal A}}^{\cal A}, \ 
R^{\cal B}, B^{\cal B}, G^{\cal B} )}_{\cal B} \models\phi.
$$
\end{small}
The interpretations of $\varepsilon$, for structures $\cal B$ that satisfy
$\phi$, are exactly the proper 3-colourings of $\cal G$.
\end{example}

\ignore{
\begin{example} Separation of instances from problem descriptions is an important feature of our framework.
For the Business Process Planner example from the Introduction, instances may
include, for example, particular kinds of goods
to be shipped, maps of inter-city roads (graphs), particular shipping deadlines, types of trucks to be used, etc.. One specification is re-used for multiple problem instances.
An example of such a specification for a particular variant of Example \ref{ex:BPP} is given in the Appendix.
\end{example}
}


The model expansion task is very common in declarative programming, -- given an input, we want to generate a solution
to a problem specified declaratively. This is usually done through grounding, i.e., combining instance structure $\cal A$ to 
a problem description $\phi$ thus obtaining a reduction to a low-level solver language such as SAT, ASP, SMT, etc. 
Model Expansion framework was
introduced for systematic study of declarative languages.
In particular, it connects KR with descriptive complexity \cite{Immerman}.
It focuses on problems, not on problem instances, it {\em separates instances from problem descriptions}.
Using the MX framework, one can produce expressiveness and capturing results for specification languages to guarantee:

\begin{compactitem}
 \item universality of a language for a class of problems,
 
 \item feasibility of a language by bounding resources
needed to solve problems in that language.
\end{compactitem}

In terms of complexity, MX lies in-between model checking (MC) (a full structure is given) and satisfiability (SAT)
(we are looking for a structure). Model generation ($\sigma = \emptyset$) has the same complexity as MX.
The authors of  \cite{Kolokolova:complexity:LPAR:long} studied the complexity
of the three tasks, MC, MX and SAT, for several logics.
Despite the importance of MX task in several research areas, the task
has not yet been studied sufficiently, unlike the two related tasks of MC and SAT.

\subsubsection{General Research Goal: Adding Modularity} Given the importance of combining different languages and solvers 
to achieve ease of axiomatization and the best performance, our goal is to
 {\em  extend the MX framework to combine modules specified in different
languages. }
The following example illustrates what we are aiming for. 

\begin{example}[Factory as Model Expansion]
In Figure \ref{fig:factory},  a part of a simple factory is represented as a modular system. 
 Both the office  and the workshop modules can be
viewed as model expansion tasks. The  instance vocabulary of the workshop is
$\sigma=\{RawMaterials\}$ and expansion vocabulary $\varepsilon=\{R\}$. 
\begin{figure}
  \begin{center}
    \includegraphics[scale=0.57]{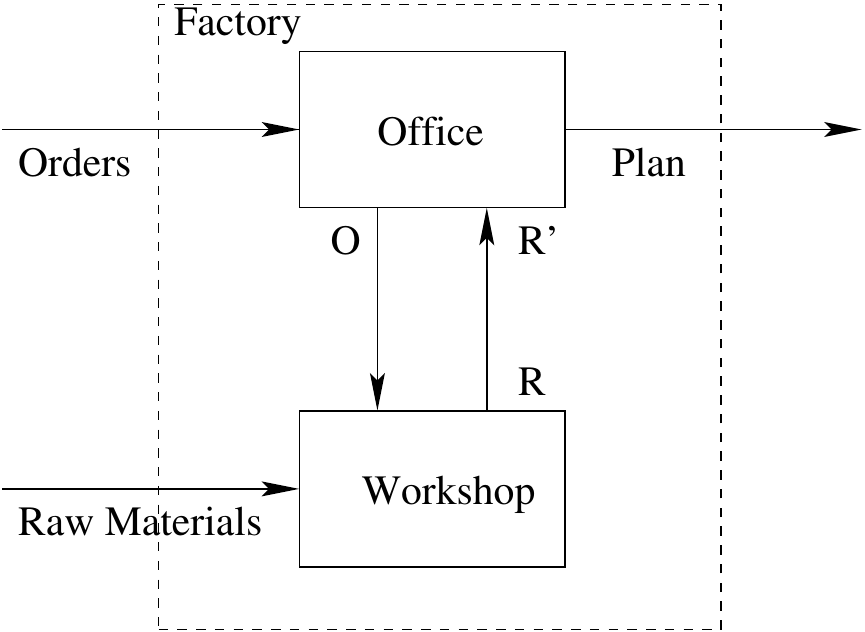}
  \end{center}
  \caption{Modular representation of a factory}
  \label{fig:factory}
\end{figure}
The bigger box with dashed borders is an MX task with instance vocabulary $\sigma'=\{Orders,RawMaterials\}$ and
expansion vocabulary $\varepsilon'=\{Plan\}$ (the ``internal'' expansion symbols $O$ and $R$ are hidden from the outside). This task is a compound MX task whose
result depends on the internal work of the office and the workshop, both of which can also have an internal structure
and be represented as modular systems themselves.
\end{example}

\subsubsection{Contributions of this paper}

In this paper, we further develop the framework of Modular Systems. 
In this framework, primitive modules represent individual knowledge bases, agents,
companies, etc. They can be axiomatized in a logic, be legacy systems, or be represented by a human 
who makes decisions. 
Unlike the previous work, we precisely define the notion of a well-formed modular system, and clearly 
separate the syntax of the algebraic language and the semantics of the algebra of modular systems.
The syntax of the algebra uses a few operations,
each of them (except feedback) is a counterpart of an operation in Codd's relational algebra, but over sets of structures rather than 
tables, and with directionality taken into account. 
The semantics of both primitive and compound modules is simply a set (class) of structures (an MX task).
By relying on the semantics of the algebra, we then introduce its natural counterpart in logic.
The logic for modular systems  allows for multiple logics axiomatizing  
individual modules in the same formula. We expect that multi-language formalisms 
such as ID-logic \cite{Denecker/Ternovska:2008:TOCL-long} will be shown to be particular instances of this logic, and other combinations 
of languages will be similarly developed. 

After giving the model-theoretic semantics of the algebra of modular systems, we define what it means, 
for a primitive module, to act as a non-deterministic operator on states of the world represented by structures
over a large vocabulary. For each expansion, there is a transition
to a new structure where the interpretation of the expansion changes, and everything else
moves to a new state by inertia. This definition is new and is more general than the one we introduced
in the previous work. We then define the semantics of the algebraic operators by Plotkin-style structural operational semantics
\cite{Plotkin}. This definition also new. We then prove the equivalence of the two semantics, operational and model-theoretic.
To illustrate the power of the projection operation, we show how a deterministic polytime program can be ``converted'' to a non-deterministic one  that solves an NP-complete problem. In general, adding projection produces a jump in the computational
complexity of the framework, similarly to feedback and union.

The authors of \cite{LT:PADL:2014} recently introduced an abstract modular inference systems formalism,
and shown how propagations in solvers can be analyzed using abstract inference rules they introduced.
We believe it is an important work. In this paper, we show how inference system can be lifted and integrated with our Modular Systems framework. The advantage of this integrations is that, with the help of the inference semantics,
 we  can now go into much greater level of details of propagation processes in our abstract algorithm 
 for solving modular systems. The  inference semantics  is the third semantics of modular systems mentioned in the title.

\ignore{One of the advantages of the MS framework is its simplicity.
 Each module is simply an MX task. There are just a few operations,
each of them (except feedback) is counterpart of an operation in Codd's relational algebra, but over sets of structures rather than 
tables, and with directionality taken into account. 
}

\subsubsection{The importance of abstract study of modularity}
We now would like to discuss the potential implications of abstract study of modularity for KR and declarative programming.

\subsubsection{A  family of multi-language KR 
formalisms} The Modular Systems framework gives rise to a  whole new family of KR formalisms by 
giving  the semantics to the combination of modules. This is can be viewed, for example, as a significant extension of answer set programming (ASP). In the past, combining  ASP programs that were created separately from each other was only possible, under some conditions,   in sequence. Now, we can combine them  in a  loop, use projections to hide parts of the vocabularies, etc. 
The previous results remain applicable. We expect, for example, that splittable programs under stable model semantics and stratifiable programs satisfy our conditions for 
sequential compositions of modules. 
Previously, in ASP, all modules had to be interpreted under one semantics (e.g. stable model semantics). Now, any model-theoretic semantics of individual modules is allowed. For example, some of the modules can be axiomatized, say, in first-order logic.
That is, in particular, our proposal amounts to  a ``modular multi-language ASP''.

\subsubsection{ Foundations in model theory} We believe that classic model theory is the right abstraction tool and a good common ground for combining formalisms developed in different communities. It is sufficiently general  and provides a rich machinery 
developed by generations of researchers.  The machinery includes, for example, deep connections between expressiveness and computational complexity.
In addition, the notion of a structure is important in KR as it abstractly 
represents our understanding of the world.


We believe that, despite common goals, the interaction between the CP community and various solver communities on one hand
and the KR community is insufficient, and that foundations in model theory can make the interaction much more 
easy and fruitful.

{\bf Analyzing other KR systems} Just as in the case of single-module system where we can use the 
purely semantical framework of model expansion, we can use the framework of Modular Systems 
to analyze multi-language KR formalisms and to study  the expressive power of modular systems.

The modular framework generalizes naturally to the case where we need to study languages (logics)
with ``built-in'' operations. In that case, embedded model expansion has to be considered,
where the embedding is into an infinite structure interpreting, e.g., built-in arithmetical 
operations \cite{TM:IJCAI:2009,TT:arithmetic:LPAR-17}.

\subsubsection{Operational  View}

Due to structural operational semantics, a new type of behaviour equivalence (bisimulation)
can be defined on complex modules (e.g. represented by ASP programs). 
The  operational view  enables us to obtain results about our modular systems such as approximability of a
sub-class of modular systems. 
While this operational view is novel and we have not developed it very much, 
we believe that  this view allows  one to apply the extensive research on proving properties of transition systems 
and the techniques developed in the situation calculus to prove useful facts about transition systems.
We can do e.g. verification of correct behaviour, static or dynamic, particularly in the presence of arithmetic.
The mathematical abstraction we proposed allows one to approach solving the 
problem of synthesis of modular systems abstractly, similarly to \cite{DBLP:journals/ai/GiacomoPS13} Just as a Golog 
program can be synthesized from  a library of available programs, a modular system
can be synthesized from a library of available solutions to MX tasks.

\subsubsection{Related Work} 

Our work on modularity was initially inspired by \cite{JOJN} who developed a constraint-based modularity formalism, where 
modules were represented by constraints and combined through operations of sequential composition and projection. 
A detailed comparison with that work is given in \cite{TT:FROCOS:2011-long}.

The connections with the  related formalism of Multi-Context Systems (MCSs), see \cite{Brewka-2007:long} 
and consequent papers, has been formally studied in \cite{Shahab-thesis}
and \cite{TT:KR2014}. We only mention here that while the contexts are very general, and may have 
any semantics, not necessarily model-theoretic, the communication between knowledge bases happens through
rules of a specific kind, that are essentially rules of logic programs with negation as failure.
We,  on the other hand, have chosen to represent communication simply through equality of vocabulary 
symbols, and to develop a model-theoretic algebra of modular systems.

Splitting results in logic programming (ASP) give conditions for 
separating a program  into modules \cite{Splitting:Turner:96,Splitting:Turner:96}.
The results rely on a specific semantics, but can be used for separating programs into modules to represent in our formalism.
The same applies to modularity of inductive definitions  \cite{Denecker/Ternovska:2008:TOCL-long,algebraic-split:VGD,DT:KR2004}.

The Generate-Define-Test parts of Answer Set Programs, as discussed in  
\cite{DBLP:conf/iclp/DeneckerLTV12}, are naturally representable as a sequential composition of the corresponding modules. 

A  recent work is \cite{LT:PADL:2014}, where the authors introduce an abstract 
approach to modular inference systems and solvers was already mentioned, and is used in this paper.

\ignore{ From DeGiacomo et al paper on service compositions G. De Giacomo et al. / Artificial Intelligence 196 (2013) 106Ð142
Most results presented in this paper appeared at an earlier stage in [30,79,15,80,26]. Here we revise, extend, and combine
them into a uniform and in-depth investigation which includes all the technical details and extended examples, so as to provide a fully-comprehensive and clear analysis of the problem and of our solution approach. In particular the technical contributions include:
¥ a notion of composition in the presence of partially controllable behaviors;
¥ a notion of composition in the presence of partially controllable behaviors;
¥ a simulation-based technique working with partially controllable behaviors, which produces ÒuniversalÓ solutions, i.e., ones from which all possible solutions can be generated;
¥ repair procedures to incrementally refine and adapt an existing solution to various unexpected types of failures;
¥ an alternative, equivalent, solution technique based on safety games well suited for model checking based technology; and a proof-of-concept implementation of the latter in the tlv system. The
}

\section{The Algebra of Modular Systems}
\label{sec:modular}

\ignore{ 
The initial
development was motivated by \cite{JOJN}, however MS framework offers two
equivalent semantics based on model-theoretic and operational views (none of
them present in \cite{JOJN}). Also, as we explain in this section, MS framework
extends that earlier work in several significant ways and, in particular, by
adding an expressive feedback operator. In the following, we first formally
define an algebraic language for modular systems and then two (equivalent)
semantics for this language. The first semantics is model-theoretical and
recursively defines the set of structures that a modular system abstractly
represents. The second semantics is called the operational semantics and associates
a non-deterministic operator to each modular system. We prove that these two
semantics coincide (i.e., the second semantics associates an operator to a
modular system whose fixpoints are closely tied to the structures that are
associated to the modular system by the first semantics). This way, we have two
distinct methods to study the properties of modular systems, and, as we shall
see, both these methods are advantageous in certain situations.
} 

 {\em Each modular system abstractly represents
an MX task}, i.e., a set (or class) of structures over some instance (input) and
expansion (output) vocabulary. Intuitively, a modular system is described as a set
of primitive modules (individual MX tasks) combined using the operations of: 
\begin{compactenum}
  \item Projection($\pi_\nu(M)$) which restricts the vocabulary of a module.
Intuitively, the projection operator on $M$ defines a modular system that acts
as $M$ internally but where  some
vocabulary symbols are hidden from the outside.
  \item Composition($M_1 \rhd M_2$) which connects outputs of $M_1$ to inputs of
$M_2$. As its name suggests, the composition operator is intended to take two
modular systems and defines a multi-step operation by serially composing $M_1$
and $M_2$.
  \item Union($M_1 \cup M_2$) which, intuitively, models the case when we have
two alternatives to do a task (that we can choose from).
  \item Feedback($M[R=S]$) which connects output $S$ of $M$ to its inputs $R$.
  As the name suggests, the feedback operator models systems with feedbacks or loops.
  Intuitively, feedbacks  represent fixpoints (not necessarily minimal) of modules viewed as operators,
  since they state that some outputs must be equal to some inputs.
  \item Complementation($\overline{M}$) which does ``the opposite'' of what $M$ does. 
\end{compactenum}
These operations are similar to the operations of Codd's relational algebra, but they work on sets 
of structures instead of relational tables. Thus, our algebra can be viewed as a higher-order 
counterpart of Codd's algebra, with loops.
One can introduce other operations, e.g. as combinations of the ones above.
The algebra of modular systems is formally defined recursively starting from primitive
modules.
\begin{definition}[Primitive Module]\label{def:primitive-mod}
A {\em primitive module} $M$ is a model expansion task (or, equivalently, a
class of structures) with distinct instance  (input)  vocabulary $\sigma$ and expansion (output) vocabulary
$\varepsilon$.  
\end{definition}

A primitive module $M$ can be given, for example, by a 
decision procedure $D_M$  that decides membership in $M$. It can also  be  given by a first- or 
 second-order formula $\phi$. In this case, $M$ is all the models of $\phi$, $M=Mod(\phi)$. 
 It could also be given by an ASP program. In this case, $M$ would be the stable models of the program,
 $M=StableMod(\phi)$.
 
 \begin{remark} \label{remark:MS-vocabularies}
 A module M can be given through axiomatizing it by a formula $\phi$ in some logic ${\cal L}$ such
that $vocab(\phi) = \sigma \cup \varepsilon_a \cup \varepsilon$. That is, $\phi$  may contain auxiliary expansion symbols that are different from
the output symbols $\varepsilon$ of $M$. (It may not even be possible to axiomatize $M$ in that particular logic
${\cal L}$ without using any auxiliary symbols). In this case, we take 
 $M = Mod(\phi)|_{(\sigma \cup \varepsilon)}$, the models of $\phi$ 
restricted to $\sigma \cup \varepsilon$.
 \end{remark}

\begin{example}
 For example, formula $\phi$ of
Example \ref{ex:colouring} describes the model expansion task for the problem of
Graph 3-colouring. Thus, $\phi$ can be the representation of a module $M_{col}$ with
instance vocabulary $\{E\}$ and expansion vocabulary $\{R,G,B\}$. 
\end{example}

Before recursively defining our algebraic language, we have to define composable
and independent modules \cite{JOJN}:
\begin{definition}[Composable, Independent]
Modules $M_1$ and $M_2$ are {\em composable} if $\varepsilon_{M_1} \cap
\varepsilon_{M_2}=\emptyset$ (no output interference). Module $M_2$ is {\em
independent} from $M_1$ if $\sigma_{M_2} \cap \varepsilon_{M_1}=\emptyset$ (no
cyclic module dependencies).
\end{definition}
Independence is needed for the definition of union, both properties, comparability and independence 
are needed for sequential composition, non-empty $\sigma$ is needed for feedback.

\begin{definition}[Well-Formed Modular Systems ($\MS(\sigma, \varepsilon)$)]\label{def:modular-system}
 The {\em set of all well-formed modular systems
$\MS(\sigma, \varepsilon)$} for a given input, $\sigma$, and output, $\varepsilon$, vocabularies is defined as follows.

\begin{enumerate}
  \item[{\bf Base Case, Primitive Modules:}] If $M$ is a primitive module with
instance (input) vocabulary $\sigma$ and expansion (output) vocabulary
$\varepsilon$, then $M \in \MS(\sigma, \varepsilon)$.

  \item[{\bf Projection}] If $M \in \MS(\sigma, \varepsilon)$ and $\tau \subseteq
\sigma \cup \varepsilon$, then $\pi_{\tau}(M) \in \MS(\sigma \cap \tau,
\varepsilon \cap \tau)$.

  \item[{\bf Sequential Composition:}] If $M \in \MS(\sigma, \varepsilon)$, $M'
\in \MS(\sigma', \varepsilon')$, $M$ is composable (no output interference) with
$M'$, and $M$ is independent from $M'$ (no cyclic  dependencies) then $(M \rhd
M') \in \MS(\sigma \cup (\sigma' \setminus \varepsilon), \varepsilon \cup \varepsilon')$.

  \item[{\bf  Union:}] If $M \in \MS(\sigma, \varepsilon)$, $M' \in \MS(\sigma', \varepsilon')$,
$M$ is independent from $M'$, and $M'$ is also independent from $M$ then $(M \cup
M') \in \MS(\sigma \cup \sigma', \varepsilon \cup \varepsilon')$.

  \item[{\bf Feedback:}] If $M \in \MS(\sigma, \varepsilon)$, $R \in \sigma$, $S \in \varepsilon$,
and $R$ and $S$ are symbols of the same type and arity, then $M[R=S] \in \MS(\sigma
\setminus \{R\}, \varepsilon \cup \{R\})$.
\item[{\bf Complementation:}]   If $M \in \MS(\sigma, \varepsilon)$, then $\overline{M} \in \MS(\sigma, \varepsilon )$.  
\end{enumerate}

Nothing else is in the set $\MS(\sigma, \varepsilon)$. 

\end{definition}
Note that the feedback (loop) operator is not defined for the case $\sigma=\emptyset$. However, composition with a module that selects structures where interpretations of two expansion predicates are equal is always possible. The feedback operator was introduced because loops are important in
information propagation, e.g. in all software systems and in solvers (e.g. ILP, ASP-CP, DPLL(T)-based)
\cite{TWT:WLP:2011,TWT:WLP-INAP:2012}. Feedback operation converts an instance predicate to an expansion predicate,
and  equates it to another expansion predicate. Feedbacks are, in a sense, fixpoints, not necessarily minimal\footnote{Modular systems under supported semantics \cite{Shahab-thesis}
allow one to focus on minimal models.}.
They add expressive power to the algebra of modular systems through  introducing additional non-determinism, which is not achieved by equating two 
expansion predicates. We discuss this issue again after the multi-language logic of modular systems 
is introduced. 

\ignore{
\begin{definition}[Modular Systems]\label{def:modular-system}
The set of {\em modular systems with instance/input vocabulary $\sigma$ and
expansion/output vocabulary $\varepsilon$}, denoted by $\MS(\sigma, \varepsilon)$,
is built recursively from primitive modules using projection, composition, union
and feedback operators. $\MS(\sigma, \varepsilon)$ is the unique minimal set
that satisfies the following conditions:

\begin{enumerate}
  \item If $M$ is a primitive module (a set or class of structures over
vocabulary $\sigma \cup \varepsilon$ along with its associated decision
procedure) with instance (input) vocabulary $\sigma$ and expansion (output)
vocabulary $\varepsilon$, then $M \in \MS(\sigma, \varepsilon)$.

  \item If $M \in \MS(\sigma, \varepsilon)$ and $\tau \subseteq \sigma \cup \varepsilon$,
then $\pi_{\tau}(M) \in \MS(\sigma_M \cap \tau, \varepsilon_M \cap \tau)$.

  \item If $M \in \MS(\sigma, \varepsilon)$, $M' \in \MS(\sigma', \varepsilon')$,
$M$ is composable (no output interference) with $M'$, and $M$ is independent
from $M'$ (no cyclic module dependencies) then $(M \rhd M') \in \MS(\sigma \cup
(\sigma' \setminus \varepsilon), \varepsilon \cup \varepsilon')$.

  \item If $M \in \MS(\sigma, \varepsilon)$, $M' \in \MS(\sigma', \varepsilon')$,
$M$ is independent from $M'$, and $M'$ is also independent from $M$ then $(M \cup
M') \in \MS(\sigma \cup \sigma', \varepsilon \cup \varepsilon')$.

  \item If $M \in \MS(\sigma, \varepsilon)$, $R \in \sigma$, $S \in \varepsilon$,
and $R$ and $S$ are symbols of the same type and arity, then $M[R=S] \in \MS(\sigma
\setminus \{R\}, \varepsilon \cup \{R\})$.
\end{enumerate}
For a modular system $M \in \MS(\sigma, \varepsilon)$, we sometimes use $\sigma_M$
to denote $\sigma$ and $\varepsilon_M$ to denote $\varepsilon$.
\end{definition}
}

\noindent The input-output vocabulary of  module $M$  is denoted $vocab(M)$. 
Modules can have ``hidden'' vocabulary symbols, see Remark \ref{remark:MS-vocabularies}.

The description of a modular system (as in Definition \ref{def:modular-system})
gives an algebraic formula representing a system.  {\em
Subsystems} of a modular system $M$ are sub-formulas of the formula that
represents $M$. Clearly, each subsystem of a modular system is a modular system
itself.

\ignore{ 
Further operators for combining modules can be defined as combinations of basic
operators above. For instance, \cite{JOJN} introduced $M_1 \blacktriangleright
M_2$ (composition with projection operator) as $\pi_{\sigma_{M_1} \cup
\varepsilon_{M_2}}(M_1 \rhd M_2)$, i.e., serial composition of $M_1$ and $M_2$
with the intermediate results (generated by $M_1$) forgotten. Also, for $M_1 \in
\MS(\sigma_1, \varepsilon_1)$ and $M_2 \in \MS(\sigma_2, \varepsilon_2)$, $M_1
\cap M_2$ is defined to be equivalent to $M_1 \rhd M_2$ (or $M_2 \rhd M_1$) when
$\sigma_1 \cap \varepsilon_2 = \sigma_2 \cap \varepsilon_1 = \varepsilon_1 \cap
\varepsilon_2 = \emptyset$, i.e., $M_1 \cap M_2$ denotes the composition of two
mutually independent components in a system.
}

\begin{example}[Simple Modular System]\label{example:simple-modular-system}
Consider the following axiomatizations of modules\footnote{In  realistic examples, 
module axiomatizations are much more complex and contain multiple rules or axioms.}, each in the corresponding logic ${\cal L}_i$. 
$$
P_{M_1}: = \{ {\cal L}_{\it WF}: \ a \leftarrow b\},
$$
$$
P_{M_2}: = \{ {\cal L}_{\it WF}: \  a \leftarrow c\},
$$
 $$
P_{M_3}: = \{ {\cal L}_{\it SM}: \  d  \leftarrow not \ a\},
$$
$$
P_{M_4}: = \{  {\cal L}_P: \ b'\lor c'  \equiv  \neg \ d\}.
$$
$ {\cal L}_{\it WF}$ is the logic of logic programs under the well-founded semantics,
$ {\cal L}_{\it SM}$  is the logic of logic programs under the stable model semantics,
$ {\cal L}_{\it P}$ is propositional logic.

The modular system in Figure \ref{fig:simple-modular-system-1} is represented by the following 
algebraic
specification.
$$
M := \pi_{\{a,b,c, d\}}(( ((M_1 \cup M_2)  \rhd M_3) \rhd M_4   ) [c=c'][b=b']).
$$

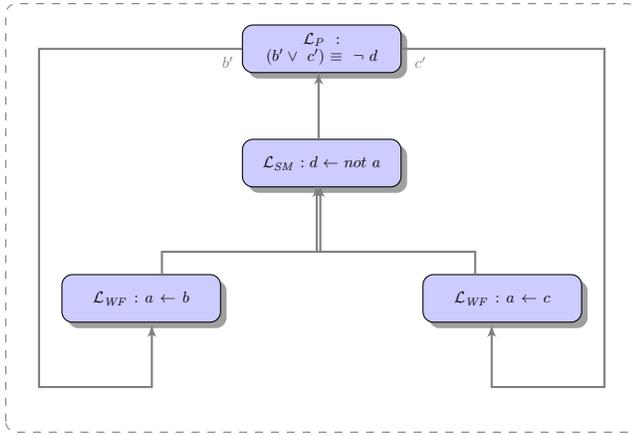
\begin{figure}
  \centering
\begin{tikzpicture}[scale=0.6,transform shape]
 
  
  
 \path \primitivemodule {M4}{${\cal L}_{\it P}:  \ ~ (b'\lor ~ c' ) \equiv ~ \neg \ d$};
 \path (modM4.south)+(0,-2) \primitivemodule{M3}{${\cal L}_{\it SM}:  d  \leftarrow not \ a$};
 \path (modM4.south)+(+4.0,-5) \primitivemodule{M2}{${\cal L}_{\it WF}:  a\leftarrow c$};
 \path (modM4.south)+(-4.0,-5) \primitivemodule{M1}{${\cal L}_{\it WF}:  a\leftarrow b$};
    

  \path [line] ($(modM1.north west)!0.63!(modM1.north east)$) -- +(-0,.5) -- +(3.51,.5) -- +(3.51, 2.) node[left=8pt, above=1pt] {} ;
  \path [line] ($(modM2.north west)!0.33!(modM2.north east)$) -- +(-0,.5) -- +(-3.51,.5) -- +(-3.51, 2) node[left=8pt, above=1pt] {} ;
    \path [line] ($(modM3.north west)!0.48!(modM3.north east)$) -- +(0, 1.5) node[right=8pt, above=1pt] {} ;
    
      \path [line] (modM4.west) node[left=9pt, below=1pt]{$b'$} -- +(-1,0) -- +(-4.5,0) -- +(-4.5, -7.5) -- +(-2, -7.5) -- +(-2, -6.1) node[left=9pt, above=1pt]{};
      
       \path [line] (modM4.east) node[right=12pt, below=1pt]{$c'$} -- +(1,0) -- +(4.5,0) -- +(4.5, -7.5) -- +(2, -7.5) -- +(2, -6.1) node[left=9pt, above=1pt]{};

  \begin{pgfonlayer}{background}
    \coordinate (A) at (-7, -8.5);
    \coordinate (B) at (7, 1);
    \path[rounded corners, draw=black!50, dashed]
   (A) rectangle (B);
 \end{pgfonlayer}

\end{tikzpicture}
  \caption{A simple modular system where modules are axiomatized in different languages.}\label{fig:simple-modular-system-1}
\end{figure}

 Module $M':= ( ((M_1 \cup M_2)  \rhd M_3) \rhd M_4   )$ has $\sigma_{M'}=  \{b,c\}$, $\varepsilon_{M'} = \{ a,b',c',d\}$.
 After adding feedbacks, we have $M'':=  M'  [c=c'][b=b']$, which turns instance symbols $b$ and $c$ into 
 expansion symbols, so we have $\sigma_{M''}=  \emptyset$ and $\varepsilon_{M''} = \{ a,b,c,b',c',d\}$, 
 and in addition, the interpretations of $c$ and $c'$,  and $b$ and $b'$ must coincide. 
 Finally, projection hides $c'$ and $b'$.

Module $M$ corresponds to the whole modular
system denoted by the box with dotted borders.  Its input-output vocabularies are as follows:  $\sigma_M= \emptyset$, $\varepsilon_M = \{ a,b,c,d\}$, $b'$ and $c'$ 
 are ``hidden'' from the outside. They are auxiliary expansion symbols, see Remark \ref{remark:MS-vocabularies}.

Modules $(M_1 \cup M_2)$ and  $M_3$ in this example are composable 
 (no output interference) and independent (no cyclic  dependencies),
$M_1$ and $M_2$ are independent.


 \end{example}
 
 The paper \cite{TT:FROCOS:2011-long}  contains a more applied example, of a business process planner,
  where each module represents a business partner.

\subsubsection{Multi-Language Logic of Modular Systems}
It is possible to introduce a multi-language logic of modular systems,
where formulas of different languages are combined
using conjunctions\footnote{It will be clear from the semantics that the operation $\rhd$ is commutative.} (standing for $\rhd$), disjunctions ($\cup$), existential second-order quantification ($\pi_\nu$), etc.
For example, model expansion for the following formula 
$$
\begin{array}{l}
\phi_M:= \exists b' \exists c' (((
\{ {\cal L}_{\it WF}: \ a \leftarrow b\}
\lor \{ {\cal L}_{\it WF}: \  a \leftarrow c\}) \\
\land  \{ {\cal L}_{\it SM}: \  d  \leftarrow not \ a\}  \land  \{ {\cal L}_{\it P}: \  d  \leftarrow not \ a\} )\\
 \land [ b=b' \land c=c'].
\end{array}
$$
with $\sigma_M = \emptyset $   and  $\varepsilon = \{ a,b,c,d \}$ and ``hidden'' (auxilliary, see Remark 
\ref{remark:MS-vocabularies})  vocabulary $\varepsilon_a = \{b',c'\}$
corresponds to  the modular system in Figure \ref{fig:simple-modular-system-1} from Example \ref{example:simple-modular-system}.

Feedback is a meta-logic operation that does not have a counterpart among logic connectives. Feedback does not exist for model generation ($\sigma=\emptyset$)
and increases  the number of symbols in the expansion vocabulary. In our example, 
 former instance symbols ($b$ and $c$ in this case) become expansion symbols, and become equal to 
 the outputs $b'$ and $c'$ thus forming loops.

Note also that projections (thus quantifiers) over  variables ranging over domain objects can be achieved if such variables are 
considered to be a part of the vocabularies of modules. 
In this logic, the full version of ID-logic, for example, would correspond to the case without feedbacks and all modules limited to either those axiomatized in first-order logic or definitions under well-founded semantics.
A formal study of such a multi-language logic in connection with existing KR formalisms (such as, e.g.  ID-logic, combinations such as  ASP and Description logic. etc.) is left as a future research direction.

Note that if all modules are axiomatized in second-order logic, our task is just 
model expansion for classic second-order logic that is naturally expressible by adding existential second-order quantifiers at the front. If there are multiple languages, we can 
talk about the {\em complexity of model expansion for the combined formula (or modular system) 
as a function of the expressiveness of the individual languages}, which is a study of practical importance.

\section{Model-Theoretic Semantics}
\label{sec:modular}

So far, we introduced the syntax of the algebraic
language using the notion of a well-formed modular system. Those  are
primitive modules (that are sets of structures) or are constructed inductively by the algebraic operations 
of composition, union, projection, loop. Model-theoretic semantics  associates, with each modular system, 
a set of structures. Each such structure is called a {\em model} of that modular
system. Let us assume that the domains of all modules are included in a (potentially infinite) universal domain $U$.
\begin{definition}[Models of a Modular System]
Let $M \in \MS(\sigma, \varepsilon)$ be a modular system and $\cB$ be a $(\sigma
\cup \varepsilon)$-structure. 
We construct the set $M^{mt}=Mod(M)$ of models of module $M$ under model-theoretic semantics recursively,
by structural induction on the structure of a module.
\begin{compactenum}
  \item[{\bf  Base Case, Primitive Module:}]\hspace{-1mm} $\cB$~is~a~model~of~$M$~if~$\cB\in M$.
  \item[{\bf Projection:}] $\cB$ is a model of $M := \pi_{(\sigma \cup
\varepsilon)}(M')$ (with $M' \in \MS(\sigma', \varepsilon')$) if  a
$(\sigma' \cup \varepsilon')$-structure $\cB'$ exists such that $\cB'$ is a
model of $M'$ and $\cB'$ expands $\cB$.
  \item[{\bf Composition:}] $\cB$ is a model of $M := M_1 \rhd M_2$ (with $M_1
\in \MS(\sigma_1, \varepsilon_1)$ and $M_2 \in \MS(\sigma_2, \varepsilon_2)$) if
 $\cB|_{(\sigma_1 \cup \varepsilon_1)}$ is a model of $M_1$ and
$\cB|_{(\sigma_2 \cup \varepsilon_2)}$ is a model of $M_2$.
  \item[{\bf Union:}] $\cB$ is a model of $M:=M_1 \cup M_2$ (with $M_1 \in \MS(
\sigma_1, \varepsilon_1)$ and $M_2 \in \MS(\sigma_2, \varepsilon_2)$) if  either $\cB|_{(\sigma_1 \cup \varepsilon_1)}$ is a model of $M_1$, or
$\cB|_{(\sigma_2 \cup \varepsilon_2)}$ is a model of $M_2$.
  \item[{\bf Feedback:}] $\cB$ is a model of $M:=M'[R=S]$ (with $M' \in \MS(
\sigma', \varepsilon')$) if $R^\cB=S^\cB$
and $\cB$ is model of $M'$.
\item[{\bf Complementation:}] $\cB$ is a model of $M:=\overline{M'}$ (with $M,M' \in \MS(
\sigma, \varepsilon)$) if and $\cB$ is not a model of $M'$. That is, $\overline{M'}$ denotes the 
complement of $M$ in the set of all possible $\sigma\cup \varepsilon$-structures  over the universal domain $U$. 

\end{compactenum}
Nothing else is a model of $M$.
\end{definition}
Note that, by this semantics, sequential composition is a commutative operation (we could have used $\Join$ notation), however the direction of 
information propagation  is uniquely given by the separations of the input and output vocabularies. 
Notice that it's not possible to compose
two modules in two different ways. 
If it was possible, then in the compound module 
we would had that the intersection
of the input and the output vocabularies would
not be empty, and this is not allowed.
So, we prefer to use
 $\ \rhd$ instead of  $\Join$ for both historic and mnemonic reasons, and encourage the reader to 
 write algebraic formulas in a way that corresponds to their visualizations of the corresponding modular systems.

An example illustrating the semantics of the feedback operator, as well as non-determinism introduced by this operator is given in the appendix.

The task of model
expansion for modular system $M$ takes a $\sigma$-structure $\cA$ and finds (or reports that
none exists) a $(\sigma \cup \varepsilon)$-structure $\cB$ that expands $\cA$
and is a model of $M$. Such a structure $\cB$ is a {\em solution of $M$ for
input $\cA$}.

\begin{remark}
The  semantics  does not put any
finiteness restriction on the domains of structures. Thus, 
the framework works for modules with infinite structures.
\end{remark}

\ignore{
\noindent \textbf{Comparison with \cite{JOJN} }  {\tt This will be replaced with a summary, probably in the introduction.}
The framework \cite{JOJN} is based on a set of variables ${\cal X}$ with each
$x \in {\cal X}$ having a domain $D(x)$. An {\em assignment} over a subset of
variables $X \subseteq {\cal X}$ is a function $\sigma : X \rightarrow \cup_{x
\in X} D(x)$, which maps variables in $X$ to values in their domains, i.e.,
$\sigma(x) \in D(x)$ for all $x \in X$. A constraint $C$ over a set of variables
$X$ is characterized by a set $C$ of assignments over $X$, called the {\em
satisfying assignments}. A primitive module in terms of \cite{JOJN} is a
constraint plus a signature for its input and outputs, i.e., two disjoint sets
$I, O \subseteq X$ of variables. Modular systems in \cite{JOJN} are combined
using operators of composition and projection, i.e., there is no union or
feedback operators. In this paper, the variables of \cite{JOJN} are represented
by the vocabulary symbols in $\sigma \cup \varepsilon$. Moreover, instead of
assigning values, we use structures that assign interpretations to vocabulary
symbols. Therefore, unlike \cite{JOJN} that is concerned with the task of
satisfiability checking (i.e., it tries to find an assignment that satisfies all
the constraint), we are concerned with the task of model expansion that is
suitable for modeling the whole system in a modular fashion (and not just one
input of the system).
}
\ignore{
Our notion of a constraint module  can
be viewed as a special case of a constraint \cite{JOJN}, provided they allow domains
which consist of sets of sets, such as the power set of a given universe, and sets of sets of tuples. In \cite{JOJN}, constraint modules are defined as assignments in contrast to
sets of structures. 
}

\ignore{
{\em Our next goal is to give a method to solve the MX task for a given modular
system, i.e., given a modular system $M$ and structure $\cA$, find $\cB$ in $M$
which expands $\cA$.} We find our inspiration in existing solver architectures
by viewing them at a high level of abstraction.
}



\section{Structural Operational  Semantics}
\label{sec:fixpoint-semantics}

In this section, we introduce a novel Structural Operational Semantics of modular systems.

We now  focus on potentially infinite  all-inclusive   vocabulary   $\tau$ that subsumes the vocabularies of all  modules considered. 
Thus, we always have $vocab(M) \subseteq \tau$.
\begin{definition}[State of a Modular Systems]
A $\tau$-state of a modular system $M \in \MS(\sigma, \varepsilon)$ is  a $\tau$-structure such that $(\sigma \cup
\varepsilon) \subseteq \tau$.
\end{definition}
\ignore{ 
With each modular system $M$, we associate a non-deterministic operator $ \modop{M}$ on  $\tau$-states. 
We use notation $\cB_1\modop{M}\cB_2$ (or,
equivalently, $\cB_2 \in  \modop{M} (\cB_1)$), to say that 
 {\em $\cB_2$ is a result of applying (non-deterministic) operator $M$ on
$\cB_1$}. 

Operational semantics of $M$ is given by the fixpoint of this operator. 
To help the reader to see the direction of our exposition, we proceed top-town, giving the main notion first.
\begin{definition}[Operational Semantics]\label{def:MS-operators}
Let $M \in \MS(\sigma, \varepsilon)$ and let
$\cB$ be a $\tau$-state of $M$. 
We say that
$\cB|_{vocab(M)}$ is a {\em model} of $M$ if  $\cB  \modop{M}\cB$,   that is $\cB$ a fixpoint of $M$.
\end{definition}

\begin{definition}[Modules as Non-Deterministic Operators]\label{def:modules-operators}
Let $M \in \MS(\sigma, \varepsilon)$. Relation  ``$\cB_2$ is a result of applying $\modop{M}$ on $\cB_1$'', denoted $\cB_1\modop{M}\cB_2$, (for $\tau$-states $\cB_1$ and $\cB_2$) is defined as follows.
\begin{compactitem}
  \item[{\bf Base Case, Primitive Modules:}] $\cB_1 \modop{M} \cB_2$ if \\
   $\cB_2|_{\sigma \cup \varepsilon} \in M$ and $\cB_2|_{(\tau \setminus
\varepsilon)}=\cB_1|_{(\tau \setminus \varepsilon)}$\footnote{In particular, it
means that $\cB_1$ and $\cB_2$ have the same domain.},
  \item[{\bf Projection:}] If $M := \pi_{(\sigma \cup \varepsilon)}(M')$ (with
$M' \in \MS(\sigma', \varepsilon')$) then $\cB_1 \modop{M} \cB_2$ if 
$(\sigma' \cup \varepsilon')$-structures $\cB'_1$ and $\cB'_2$ exist s.t.

  \begin{compactenum}
    \item $\cB'_1|_{(\sigma \cup \varepsilon)}=\cB_1|_{(\sigma \cup \varepsilon)}$
(expanding $\cB_1$ with interpretations of projected-out symbols gives $\cB'_1$),
    \item $\cB'_1\modop{M'}\cB'_2$ (applying $M'$ on $\cB'_1$ gives $\cB'_2$),
    
    \item $\cB'_2|_{(\sigma \cup \varepsilon)}=\cB_2|_{(\sigma \cup \varepsilon)}$
(projecting $\cB'_2$ on $\sigma \cup \varepsilon$ gives $\cB_2$), and,
    \item $\cB_1|_{\tau \setminus (\sigma \cup \varepsilon)}=\cB_2|_{\tau \setminus
(\sigma \cup \varepsilon)}$ ($M$ can only affect its vocabulary).
  \end{compactenum}
  \item[{\bf Composition:}] If $M := M_1 \rhd M_2$ (with $M_1 \in \MS(\sigma_1,
\varepsilon_1)$ and $M_2 \in \MS(\sigma_2, \varepsilon_2)$) then $\cB_1 \modop{M}
\cB_2$ if  $\tau$-structure $\cB'$ exists such that $\cB_1 \modop{M_1}
\cB'$ and $\cB' \modop{M_2}\cB_2$,
  \item[{\bf Union.}] If $M:=M_1 \cup M_2$ (with $M_1 \in \MS(\sigma_1,
\varepsilon_1)$ and $M_2 \in \MS(\sigma_2, \varepsilon_2)$) then $\cB_1 \modop{M}
\cB_2$ if either $\cB_1\modop{M_1}\cB_2$ or $\cB_1 \modop{M_2} \cB_2$,
  \item[{\bf Feedback:}] If $M:=M'[R=S]$ (with $M' \in \MS(\sigma',
\varepsilon')$) then $\cB_1 \modop{M} \cB_2$ if 
natural number $n \geq 1$ and $\tau$-structures $\cB^0_1, \cdots, \cB^n_1$ and
$\cB^0_2, \cdots,\cB^{n-1}_2$ exist such that:
$$
\begin{array}{l}
\cB^0_1 = \cB_1, \cB^n_1=\cB_2,\\
\cB^i_1\modop{M'}\cB^i_2 \mbox{ for all $i$ such that $0 \leq i < n$},\\
\cB^{i+1}_1|_{(\tau \setminus \{R\})}=\cB^i_2|_{(\tau \setminus \{R\})} \mbox{
for all $i$ such that  $0 \leq i < n$},\\
R^{\cB^{i+1}_1}=S^{\cB^i_2} \mbox{ for all $i$ such that $0 \leq i < n$},\\
R^{\cB^n_1} = S^{\cB^n_1}.
\end{array}
$$
\end{compactitem}
Nothing else is in the relation $\cB_1\modop{M}\cB_2$.
\end{definition}

\begin{figure}
\begin{center}
\begin{tikzpicture}[->,>=stealth',shorten >=1pt,auto,node distance=2cm,
  thick,main node/.style={circle,draw,font=\sffamily}]

  \node[main node] (B1) {$\cB_1$};
  \node[main node] (Bp1) [right of=B1] {$\cB'_1$};
  \node[main node] (Bp2) [right of=Bp1] {$\cB'_2$};
  \node[main node] (B2) [right of=Bp2] {$\cB_2$};

  \path [line] (B1) -- (Bp1) node [above, midway] {expand};
  \path [line] (Bp1) -- (Bp2) node [above, midway] { $M'$};
  \path [line] (Bp2) -- (B2) node [above, midway] {project};
\end{tikzpicture}
\end{center}
\caption{Module $M:=\pi_\tau(M')$ operates by (a) expanding vocabulary of input
$\cB_1$, (b) applying $M'$ to expanded input, and, (c) projecting the result of
application of $M'$.}\label{fig:MS-proj-operator}
\end{figure}

\begin{figure}
\begin{center}
\begin{tikzpicture}[->,>=stealth',shorten >=1pt,auto,node distance=1.6cm,
  thick,main node/.style={circle,draw,font=\sffamily}]

  \node[main node] (B1) {$\cB_1$};
  \node[main node] (B01) [right of=B1] {$\cB^0_1$};
  \node[main node] (B02) [below of=B01] {$\cB^0_2$};
  \node[main node] (B11) [right of=B01] {$\cB^1_1$};
  \node[main node] (B12) [below of=B11] {$\cB^1_2$};
  \node[main node, dotted] (Bm1) [right of=B11] {$\cB^i_1$};
  \node[main node, dotted] (Bm2) [below of=Bm1] {$\cB^i_2$};
  \node[main node] (Bn1) [right of=Bm1] {$\cB^n_1$};
  \node[main node] (B2) [right of=Bn1] {$\cB_2$};

  \path [line] (B1) -- (B01) node [above, midway] {copy};
  \path [line] (B01) -- node [sloped, anchor=center, below] { $M'$} (B02);
  \path [line] (B02) -- node [sloped, anchor=center, above] {$R := S$} (B11);
  \path [line] (B11) -- node [sloped, anchor=center, below] { $M'$} (B12);
  \path [line] (B12) -- node [sloped, anchor=center, above] {$R := S$} (Bm1);
  \path [line,dotted] (Bm1) -- node [sloped, anchor=center, below] {$M'$} (Bm2);
  \path [line] (Bm2) -- node [sloped, anchor=center, above] {$R := S$} (Bn1);
  \path [line] (Bn1) -- (B2) node [above, midway] {copy};
\end{tikzpicture}
\end{center}
\caption{Module $M:=M'[R=S]$ operates by repeatedly applying $M'$ on the given
structure and copying interpretation of $S$ into interpretation of $R$ until it
reaches a fixpoint of $M'$.}\label{fig:MS-feedback-operator}
\end{figure}

} 

The semantics we give  is structural
because, for example, the meaning of the sequential composition, $M_1 \rhd M_2$,
is defined through the meaning of $M_1$ and the meaning of $M_2$. 

\begin{definition}[Modules as Operators]
 We say that a well-formed modular system  $M$ (non-deterministically) maps $\tau$-state $\cB_1$ to $\tau$-state $\cB_2$,  notation $(M, \cB_1) \longrightarrow
\cB_2$,   if we can apply the rules of the structural
operational semantics (below) starting from this expression and arriving to $true$.
In that case, we say that transition  $(M, \cB_1) \longrightarrow
\cB_2$ is {\em derivable}.
\partitle{Primitive modules}  $M$:
$$
\frac{(M, \cB_1) \longrightarrow \cB_2}{true}\mbox{ if } \cB_2|_{(\sigma \cup
\varepsilon)} \in M \mbox{ and } \cB_2|_{(\tau \setminus \varepsilon)}=
\cB_1|_{(\tau \setminus \varepsilon)}.
$$
We proceed by induction on the structure of modular system $M$.
\partitle{Projection}  $\pi_\nu(M)$:
$$
\frac{(\pi_\nu(M), \cB_1) \longrightarrow \cB_2}{(M, \cB'_1) \longrightarrow
\cB'_2} \mbox{ if } \cB'_1|_\nu = \cB_1|_\nu \mbox{ and } \cB'_2|_\nu=\cB_2|_\nu.
$$

\partitle{Composition} $M_1 \rhd M_2$:
$$
\frac{(M_1 \rhd M_2, \cB_1) \longrightarrow \cB_2}{(M_1, \cB_1) \longrightarrow
\cB' \mbox{ and } (M_2, \cB') \longrightarrow \cB_2}.
$$

\partitle{Union}  $M_1 \cup M_2$:
$$
\frac{(M_1 \cup M_2, \cB_1) \longrightarrow \cB_2}{(M_1, \cB_1) \longrightarrow
\cB_2},
\ \ \ \ \ \ \ \ \ 
\frac{(M_1 \cup M_2, \cB_1) \longrightarrow \cB_2}{(M_2, \cB_1) \longrightarrow
\cB_2}.$$

\partitle{Feedback} $M[R=S]$:
$$
\frac{(M[R=S], \cB_1) \longrightarrow \cB_2}{(M, \cB_1) \longrightarrow \cB_2}, 
\mbox{ if } R^{\cB_1} = S^{\cB_2}.
$$

\partitle{Complementation} $\overline{M}$:
$$
\frac{(\overline{M}, \cB_1) \longrightarrow \cB_2}{true}  \mbox{ if $(M, \cB_1) \longrightarrow \cB_2$ is not derivable}.
$$
Nothing else is derivable.
\end{definition}

Let us clarify the projection operation $\pi_\nu(M)$.  
 Let $vocab(M) = \sigma'\cup \varepsilon'$, let 
 $\nu=\sigma\cup\varepsilon$, $\sigma \subseteq \sigma'$, $\varepsilon \subseteq \varepsilon'$.
Module $\pi_\nu(M)$, viewed as an operator, is applied to $\tau$-structure $\cB_1$. It (a) expands $\sigma$-part of $\cB_1$ to $\sigma'$  by an arbitrary interpretation over the same domain, and then  (b) applies $M$ to the modified input, (c) projects the result of application of $M$ onto $\varepsilon$, ignoring everything else, (d) the interpretations of $\tau\setminus \epsilon$ are moved from $\cB_1$ by inertia.

\begin{definition}[Operational Semantics] \label{def:operational-semantics}
Let $M$ be a well-formed modular system in  $\MS(\sigma, \varepsilon)$.
The semantics of $M$ is given by the following set.
$$
M^{op} := \{ \cB \  |\   (\cB_1, M) \longrightarrow \cB_2   \mbox{ and }   \cB|_\sigma = \cB_1|_\sigma, \ \cB|_\varepsilon =
\cB_2|_\varepsilon  \}.
$$
\end{definition}
Figure \ref{fig:MS-state-change} illustrates this definition.

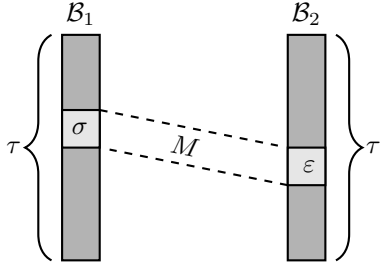
\begin{figure}
\begin{center}
\begin{tikzpicture}[thick]

  \node (R1A) {};
  \path (R1A)+(0.5,-3) node (R1B) {};
  \draw [fill=black!30] (R1A) rectangle (R1B);

  \path (R1A)+(0,-1) node (R1sA) {};
  \path (R1sA)+(0.5,-0.5) node (R1sB) {};
  \draw [fill=black!10] (R1sA) rectangle (R1sB);

  \path (R1A)+(3,0) node (R2A) {};
  \path (R1B)+(3,0) node (R2B) {};
  \draw [fill=black!30] (R2A) rectangle (R2B);

  \path (R1sA)+(3,-0.5) node (R2eA) {};
  \path (R1sB)+(3,-0.5) node (R2eB) {};
  \draw [fill=black!10] (R2eA) rectangle (R2eB);

  \path (R1sA)+(0.5,0) edge[-,dashed] node[sloped,anchor=center,below]{$M$} (R2eA);
  \path (R2eA)+(0,-0.5) edge[-,dashed] (R1sB);

  \path (R1sA)+(0,-0.25) node[right] {$\sigma$};
  \path (R2eB)+(0,0.25) node[left] {$\varepsilon$};

  \draw [decorate,decoration={brace,amplitude=10pt},xshift=-4pt,yshift=0pt] (0,-3) -- (0,0) node [black,midway,xshift=-0.5cm] {$\tau$};
  \draw [decorate,decoration={brace,amplitude=10pt},xshift=4pt,yshift=0pt] (3.5,0) -- (3.5,-3) node [black,midway,xshift=0.5cm] {$\tau$};

  \path (0.25,0) node [above] {$\cB_1$};
  \path (3.25,0) node [above] {$\cB_2$};
\end{tikzpicture}
\end{center}
\caption{An illustration of Definition \ref{def:operational-semantics}. Module $M \in \MS(\sigma, \varepsilon)$ maps a $\tau$-structure $\cB_1$
(with $ (\sigma \cup \varepsilon)\subseteq \tau$) to a $\tau$-structure $\cB_2$
by changing the interpretation  $\varepsilon$ according
to  $M$ (so that the $\sigma$ part and the new $\varepsilon$ part,
together, form a model of $M$). Interpretation of all other symbols, including
those in $\sigma$, stays the same.
 This is similar to how
frame axioms keep fluents that are not affected by actions unchanged in 
the situation 
calculus. 
}\label{fig:MS-state-change}
\end{figure}
\begin{corollary}
Every result of application of $M$ is its fixpoint. That is, 
for any $\tau$-states $\cB_1$, $\cB_2$, if $(M, \cB_1) \longrightarrow \cB_2$,
then $(M, \cB_2) \longrightarrow \cB_2$.
\end{corollary} 
\begin{proof} By Definition \ref{def:operational-semantics}, because of inertia, the interpretation of $\sigma$ is transferred from 
$\cB_1$ to $\cB_2$. Since the interpretation of $\varepsilon$ is already changed by $M$, nothing is to be changed, and $(M, \cB_2) \longrightarrow \cB_2$. 
\end{proof}


\begin{theorem}[Operational = Model-theoretic Semantics]\label{thm:fp-mt-equiv}
Let $M$ be a well-formed modular system in  $\MS(\sigma, \varepsilon)$. Then, its 
model-theoretic and operational semantics coincide, 
$$
M^{mt}=M^{op}.
$$
\end{theorem}
The most important consequence of this theorem is that all the
results obtained when modules are viewed as operators, still hold when modules
are viewed as sets of structures (and vice versa). Thus, we may
use either of these semantics.
From now on, by $M$ we mean either one of these sets $M^{mt}$ or $M^{op}$.

 \begin{proof}   We  prove the statement inductively.

\partitle{Base case, primitive module} By definition, model-theoretically, $\cB$ is a model of $M$ if $\cB\in M$.
On the other hand, operationaly,
$$
M^{op} := \{ \cB \  |\   (\cB_1, M) \longrightarrow \cB_2   \mbox{ and }   \cB|_\sigma = \cB_1|_\sigma, \ \cB|_\varepsilon =
\cB_2|_\varepsilon  \},
$$
where
$$
\frac{(M, \cB_1) \longrightarrow \cB_2}{true}\mbox{ if } \cB_2|_{(\sigma \cup
\varepsilon)} \in M \mbox{ and } \cB_2|_{(\tau \setminus \varepsilon)}=
\cB_1|_{(\tau \setminus \varepsilon)}.
$$
Thus, $\cB\in M$, and the two semantics coincide for primitive modules.

Our inductive hypothesis is that the statement of the theorem 
holds for $M_1$, $M_2$ and $M'$. We proceed inductively. 

\partitle{Projection}  $M:=\pi_\nu(M')$. By the hypothesis, $(M')^{mt}=(M')^{op}$, where $(M')^{op}$ is constructed ``from pieces'',
$
(M')^{op} := \{ \cB' \  |\   (\cB'_1, M') \longrightarrow \cB'_2   \mbox{ and }   \cB|_\sigma = \cB'_1|_\sigma, \ \cB|_\varepsilon =
\cB'_2|_\varepsilon  \}.
$
We apply the rule
$$
\frac{(\pi_\nu(M'), \cB_1) \longrightarrow \cB_2}{(M', \cB'_1) \longrightarrow
\cB'_2} \mbox{ if } \cB'_1|_\nu = \cB_1|_\nu \mbox{ and } \cB'_2|_\nu=\cB_2|_\nu
$$
and obtain that $(\pi_\nu(M'), \cB_1) \longrightarrow \cB_2$ where  $\cB_1$ and $\cB_2$ are just like $\cB'_1$ and $\cB'_2$ on the vocabulary $\nu$.  Now, $M:= \pi_\nu(M')$ is constructed ``from $\sigma$ and $\varepsilon$ pieces'' of $\cB_1$ and $\cB_2$, respectively (where $\nu = \sigma \cup \varepsilon$):
$$
M^{op} := \{ \cB \  |\   (\cB_1, M) \longrightarrow \cB_2   \mbox{ and }   \cB|_\sigma = \cB_1|_\sigma, \ \cB|_\varepsilon =
\cB_2|_\varepsilon  \},
$$

On the other hand, model-theoretically, $\cB$ is a model of $M := \pi_{(\sigma \cup
\varepsilon)}(M')$ (with $M' \in \MS(\sigma', \varepsilon')$) if  a
$(\sigma' \cup \varepsilon')$-structure $\cB'$ exists such that $\cB'$ is a
model of $M'$ and $\cB'$ expands $\cB$, which makes the two semantics equal for projection, $(M)^{mt}=(M)^{op}$.

\ignore{
\partitle{Composition} $M:= M_1 \rhd M_2$. By the hypothesis, $M_1^{mt}=M_1^{op}$ and $M_2^{mt}=M_2^{op}$.
...... {\tt NOT FINISHED}

By the definition of operational semantics, if $(M_1, \cB) \longrightarrow
\cB$ and $(M_2, \cB) \longrightarrow \cB$ are derivable, then so is $(M_1 \rhd M_2, \cB) \longrightarrow
\cB$ .

 $M_1 \rhd M_2, \cB_1$,
if 
$$
\frac{(M_1 \rhd M_2, \cB_1) \longrightarrow \cB_2}{(M_1, \cB_1) \longrightarrow
\cB' \mbox{ and } (M_2, \cB') \longrightarrow \cB_2}.
$$
Model-theoretically, $\cB$ is a model of $M := M_1 \rhd M_2$ (with $M_1
\in \MS(\sigma_1, \varepsilon_1)$ and $M_2 \in \MS(\sigma_2, \varepsilon_2)$) if
 $\cB|_{(\sigma_1 \cup \varepsilon_1)}$ is a model of $M_1$ and
$\cB|_{(\sigma_2 \cup \varepsilon_2)}$ is a model of $M_2$.

\partitle{Union}  $M_1 \cup M_2$:
$$
\frac{(M_1 \cup M_2, \cB_1) \longrightarrow \cB_2}{(M_1, \cB_1) \longrightarrow
\cB_2},
\ \ \ \ \ \ \ \ \ 
\frac{(M_1 \cup M_2, \cB_1) \longrightarrow \cB_2}{(M_2, \cB_1) \longrightarrow
\cB_2}.$$

Model-theoretically, $\cB$ is a model of $M:=M_1 \cup M_2$ (with $M_1 \in \MS(
\sigma_1, \varepsilon_1)$ and $M_2 \in \MS(\sigma_2, \varepsilon_2)$) if  either $\cB|_{(\sigma_1 \cup \varepsilon_1)}$ is a model of $M_1$, or
$\cB|_{(\sigma_2 \cup \varepsilon_2)}$ is a model of $M_2$.

\partitle{Feedback} $M[R=S]$:
$$
\frac{(M[R=S], \cB_1) \longrightarrow \cB_2}{(M, \cB_1) \longrightarrow \cB_2}, 
\mbox{ if } R^{\cB_1} = S^{\cB_2}.
$$

Model-theoretically, $\cB$ is a model of $M:=M'[R=S]$ (with $M' \in \MS(
\sigma', \varepsilon')$) if $R^\cB=S^\cB$
and $\cB$ is model of $M'$.

}
We omit the proofs for the other inductive cases. 
\end{proof}

\subsubsection{Applications of Operational  View}

We now discuss how the operational semantics can be used.
For example, we can consider modular systems at various levels of granularity. We might be interested in the following question: if $M$ gives a transition from a structure $\cB$ to structures $\cB'$, then what are the transitions given by the subsystems of $M$? While answering this question in its full generality is algorithmically impossible, we may study the question of whether a particular transition by a subsystem exists. To answer it, one has to start from the system and build down to the subsystem using the rules of the structural operation semantics.
Reasoning about subsystems of a modular system can be useful in business process modelling. Suppose a 
particular transition should hold for the entire process. This might be the global task of an organization. 
In order to make that transition, the subsystems have to perform their own transitions. Those transitions are derivable using
the rules of structural operational semantics. 


\subsubsection{Complexity}
In the following proposition, we assume a standard encoding of structures as binary strings ) as is common in Descriptive complexity \cite{Immerman}. Note that if $M$ is deterministic, it is polytime in the size of the encoding of the input structure. This is because the domain remains the same, the arities of the relations in $\varepsilon$ are fixed, so we need $(n^k)$ 
steps to construct new interpretations of $\varepsilon$, and move the remaining relations.
\begin{proposition}
Let $M$ be a module that performs a (deterministic) polytime computation. 
Projection $\pi_{\nu}(M) $ increases the complexity of $M$ from P to NP. More generally, for an operator $M$ on the $k$-th level of the Polynomial Time hierarchy (PH), projection can increase the complexity of $M$ from  $\Delta^P_k$ to  $\Sigma^P_{k+1}$.
\end{proposition}
\begin{proof} We will show the property for the jump from P to NP, for illustration. The proof generalizes to all levels of PH. Let $M$ takes an instance of an NP-complete problem, such as a graph in 3-Colourability, encoded in $\sigma_G$, and what it means to be 3-Colourable, as a formula encoded in the interpretation of $\sigma_{\phi}$, and returns an instance of SAT encoded in $\varepsilon$, a CNF formula that is satisfiable if and only if the graph is 3-Colourable, and a yes/no answer bit represented by $\varepsilon_{\it answer}$. Thus, $M$ performs a deterministic (thus, polytime)
reduction. Consider $\pi_{\nu}(M) $, where $\nu = \sigma_G \cup \varepsilon_{answer}$. This module takes a graph and returns
a yes or no answer depending on whether the graph is 3-colourable. Thus, $\pi_{\nu}(M) $ solves an NP-complete problem.
\end{proof}

Union and feedback change the complexity as well. 
\ignore{
{\tt Give examples of how the sets of models change (or propositions) here.}

Later, we'll show that the modular framework bases on  Lierler and  Truszczy« nski's work does not give an increase in complexity if modules are combined.

\begin{theorem}
ref to old work about complexity, makes sense for finite domains only.
\end{theorem}
}

\ignore{
{\tt WIKIPEDIA: Operational semantics are a category of formal programming language semantics in which certain desired properties of a program, such as correctness, safety or security, are verified by constructing proofs from logical statements about its execution and procedures.

Operational semantics are classified in two categories: structural operational semantics (or small-step semantics) formally describe how the individual steps of a computation take place in a computer-based system. By opposition natural semantics (or big-step semantics) describe how the overall results of the executions are obtained. Other approaches to providing a formal semantics of programming languages include axiomatic semantics and denotational semantics.
}

{\tt WIKIPEDIA: Structural operational semantics (also called structured operational semantics or small-step semantics) was introduced by Gordon Plotkin in (Plotkin81) as a logical means to define operational semantics. The basic idea behind SOS is to define the behavior of a program in terms of the behavior of its parts, thus providing a structural, i.e., syntax oriented and inductive, view on operational semantics. An SOS specification defines the behavior of a program in terms of a (set of) transition relation(s). SOS specifications take the form of a set of inference rules which define the valid transitions of a composite piece of syntax in terms of the transitions of its components.
}

}

\section{Inference  Semantics of Modular Systems}
\label{sec:inference-semantics}

In modular systems, each agent or a knowledge base can have its own way of reasoning, that can be formulated through 
inferences or propagations.
To define inferential semantics for modular systems, we closely follow
\cite{LT:PADL:2014}.
Since input/output is not considered by the authors, their case corresponds to the instance vocabulary being empty, $\sigma = \emptyset$, i.e., model generation, and can be viewed as an analysis of the after-grounding faze. 
Since we want to separate problem descriptions  and their instances (and  reuse problem descriptions), as well as
to define additional algebraic operations (the authors consider conjunctions only), 
we need to allow $\sigma \not = \emptyset$, and present inferences on partial structures. 
This is not hard however.

We start by assuming that there is a constant for every element of the domains.
We view structures as sets of ground atoms.  
We now closely follow and generalize the definitions of  \cite{LT:PADL:2014} from sets of propositional atoms to first-order structures,
to establish a connection  to the Modular Systems framework presented above. The propositional case then corresponds to structures 
over the domain $\{ \langle \ \rangle \}$ containing the empty tuple that interprets propositional symbols that are true.

Let a  fixed countably infinite set of ground atoms $\tau$ be given.
We use $Lit(\tau)$ to denote the set of all literals over $\tau$.
For $S\subseteq Lit(\tau)$:

$S^+:= \tau \cap S$ 

$S^-:= \{ a \in \tau \ |  \ \neg a \in S\}$

$l \in Lit(\tau)$ is unassigned in $S$ if $l\not \in M$ and $\bar{l} \not \in S$ 

$ S $ is consistent if $S^+ \cap S^- \not =  \emptyset  $

Let $C(\tau)$ be all consistent subsets of $Lit(\tau)$.
\begin{definition}[Abstract Inference Representation of $M$]
An {\em abstract inference representation $M^i$ of module $M$} over a vocabulary $\tau$ is
a finite set of pairs of the form $(S, l)$, where $S\in C(\tau)$, $l \in Lit(\tau)$, and $l \not\in Lit(\tau)$.
Such pairs are called {\em inferences of the module $M$}.
\end{definition}

In the exposition below, we view structures as sets of propositional atoms, $\cB\subseteq \tau $.

$S$ is consistent with $\cB\subseteq \tau$ if $S^+ \subseteq \cB$ and $S^- \cap \cB = \emptyset$.
Literal $l$ is consistent with $\cB\subseteq \tau$ if $\{l\}$ is consistent with $\cB$.

\ignore{
\begin{definition}
Structure $\cB\subseteq \tau $ (represented as a set of ground atoms) is a model of a module $M$ if for every inference  $(S,l)\in M$ such as $S$ is consistent with $\cB$, 
$l$ is consistent with $\cB$, too. 
\end{definition}
}
\begin{definition}[Primitive Module, Inferential Semantics]
A primitive module $M\in MS(\sigma,\varepsilon)$ is a set of $(\sigma \cup \varepsilon)$-structures $\cB $  such that 
for every inference  $(S,l)\in M^i$ such as $S$ is consistent with $\cB$, 
$l$ is consistent with $\cB$, too. 
\end{definition}
Thus, primitive modules, even when they are represented 
through abstract inferences, are sets of structures as before, and the definitions of the algebraic operations 
do not need to be changed.



The inference framework can be viewed as yet another (very useful) way of representing modules. 
Since the inference framework is abstract, we cannot prove a correspondence between  a given individual module presented 
as a set of structures or as an operator on one hand and as an inferential representation on the other in general, without specifying what 
inference mechanism is used. However, we can do it for 
particular cases such as $Ent(T)$ \cite{LT:PADL:2014}, which is left for a future paper. 

With the inference semantics as described, we can now model problems (sets of instances) rather than single
instances as a combination of other problems.
This semantics allows one to study the details of propagation of information in the process of constructing solutions 
to modular systems, through incremental construction of partial structures as in \cite{TWT:WLP:2011,TWT:WLP-INAP:2012}, but in more detail.  
This direction is left for future research.

\section{Conclusion and Future Directions}

We described a modular system framework, where primitive and compound 
 modules are  sets (classes) of structures, and combinations of modules 
 are achieved by applying algebraic operations that are a higher-order counterpart
 of Codd's relational algebra operations. An additional operation is the feedback operator
 that connects output symbols with the input ones and is used to model information propagation
 such as loops of software systems and solvers.  
 
We defined two novel semantics of modular systems, operational and 
inferential, that are equivalent to the original model-theoretic semantics \cite{TT:FROCOS:2011-long}. 
We  presented a multi-language logic, a syntactic counterpart of the algebra of modular systems.
Minimal models of modular systems 
 are introduced in a separate paper on supported modular systems, 
 see also \cite{Shahab-thesis}.

The framework of modular systems gives us, through its semantic-based approach,
a unifying perspective on multi-language formalisms and solvers.
More importantly, it  gives rise to a whole new family of multi-language KR formalisms,
where new formalisms can be obtained by instantiating specific logics defining individual 
modules. 

The framework can be used for analysis of existing KR languages. In particular, 
expressiveness and complexity results for combined formalisms can be obtained 
in a way similar to the previous work \cite{MT-essence-journal:2008,TT:NonMon30,TT:arithmetic:LPAR-17} where single-module embedded 
model expansion was used.

\bibliographystyle{aaai/aaai.bst}
\bibliography{Files/bibliography}

\section{Appendix}

\begin{example}
We illustrate models of a simple modular system with feedback operator.
Consider the following axiomatization $P_{M_0}$ of a primitive module ${M_0}$, where $\sigma_{M_0} = \{ i\}$ and
$\varepsilon_{M_0} = \{ a, b \}$.
$$ P_{M_0} :=
\left \{ {\cal L}_{\it SM}: \
\begin{array}{l} 
a\leftarrow i, not \ b, \\
b\leftarrow i, not \ a. 
\end{array}
\right \}
$$
We will demonstrate how the set of models of this program changes when we use the feedback operator.
When the input $i$ is true (given by the corresponding instance structure), then 
$$StableMod(P_{M_0}, \ i= true ) = \{  \{ a \} , \{ b \} \}.$$
When  $i$ is false, there is one model, where everything is false, 
$$StableMod(P_{M_0},\  i= false) = \{  \emptyset\}.$$
Module ${M_0}$ is the set of structures for the entire $\sigma_{M_0} \cup \varepsilon_{M_0}$ vocabulary.
Since we are dealing with a propositional case, each structure is represented by a set of atoms that are true in that structure.
$$
		{M_0} = \{   \{ i,a \} , \{i, b \}, \emptyset  \}.
$$
Now consider a different module, ${M_1}$, with $\sigma_{M_1} = \{ i,  a, b \}$ and
$\varepsilon_{M_1} = \{a', b' \}$, axiomatized by 

$$ P_{M_1} :=
\left \{ {\cal L}_{\it SM}: \
\begin{array}{l} 
a'\leftarrow i, not \ b, \\
b'\leftarrow i, not \ a. 
\end{array}
\right \}
$$

\begin{figure}
  \centering
\begin{tikzpicture}[scale=0.6,transform shape]

 \path \primitivemodule {M4}{$\begin{array}{l} 
a'\leftarrow i, not \ b, \\
b'\leftarrow i, not \ a. 
\end{array}$};

 \path [line] ($(modM4.north west)!0.23!(modM4.north east)$) -- +(-0,1) node[left=9pt, above=1pt]{a'};
 
  \path [line] ($(modM4.north west)!0.63!(modM4.north east)$) -- +(-0,1) node[left=9pt, above=1pt]{b'};
  
   \path [line] ($(modM3.north west)!0.17!(modM3.north east)$) -- +(-0,1.5) node[left=9pt, below=6pt]{a};
      \path [line] ($(modM3.north west)!0.43!(modM3.north east)$) -- +(-0,1.5) node[left=9pt, below=6pt]{i};
 
  \path [line] ($(modM3.north west)!0.73!(modM3.north east)$) -- +(-0,1.5) node[left=9pt, below=6pt]{b};

 \begin{pgfonlayer}{background}
    \coordinate (A) at (-5, -3.5);
    \coordinate (B) at (5, 3);
    \path[rounded corners, draw=black!50, dashed]
   (A) rectangle (B);
 \end{pgfonlayer}

\end{tikzpicture}
  \caption{Module  $M_1$.}\label{fig:simple-modular-system}
\end{figure}
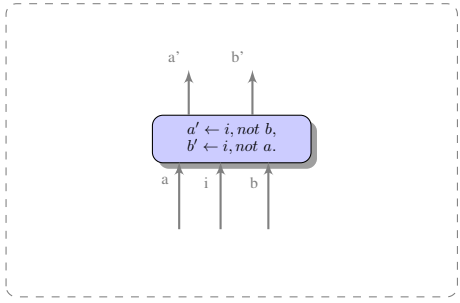

This modular system is deterministic, -- for each input (each of the eight possible interpretations of $i$, $a$ and $b$), there is at most one model. 

\begin{center}
  \begin{tabular}{| c | c | c || c | }
    \hline
    $i $& $a$ & $b$ & Models  of $M_1$\\ \hline

    $\bot$ & $ \bot$ & $ \bot$ & $\{ \emptyset \}$ \\ \hline

    $\bot$ & $ \top$ & $ \bot$ & $\{ \emptyset \}$ \\ \hline

    $\bot$ & $ \bot$ & $ \top$ & $\{ \emptyset \}$ \\ \hline

    $\bot$ & $ \top$ & $ \top$ & $\{ \emptyset \}$ \\ \hline

    $\top$ & $ \bot$ & $ \bot$ & $\{  \{ i,a',b' \} \}$ \\ \hline
   
     $\top$ & $ \top$ & $ \bot$ & $\{  \{ i,a,a' \} \}$ \\ \hline

 $\top$ & $ \bot$ & $ \top$ & $\{ \{ i,b,b' \} \}$ \\ \hline

 $\top$ & $ \top$ & $ \top$ & $\{ \{ i,a,b \} \}$ \\ \hline

    \hline
  \end{tabular}
\end{center}

Thus, we have: 
$$
M_1 = \{ \emptyset       ,  \{ i,a',b' \},      \{ i,a,a' \},  \{ i,b,b' \}  , \{ i,a,b \}           \}.
$$
If we add feedback, we obtain the following system $ M_2 = M_1 [a=a'][b=b'] $.
Its input is $i$, all other symbols are in the expansion vocabulary. The  models are:
\begin{center}
  \begin{tabular}{| c  || c | }
    \hline
    $i $&  Models  of $M_2$\\ \hline

    $\bot$ &  $\{ \emptyset \}$ \\ \hline

    $\bot$  & $\{ \emptyset \}$ \\ \hline

     $\top$ &   $\{  \{ i,a,a' \} \}$ \\ \hline

 $\top$ &  $\{ \{ i,b,b' \} \}$ \\ \hline

    \hline
  \end{tabular}
\end{center}

$$
M_2 = M_1 [a=a'][b=b'] =   \{ \emptyset       ,      \{ i,a,a' \},  \{ i,b,b' \}             \}.
$$
As we see here, after adding feedback, for the same input $i$, we obtain two different 
models. Thus, by means of feedback, a deterministic system $M_1$  was turned into a non-deterministic system $M_2$.
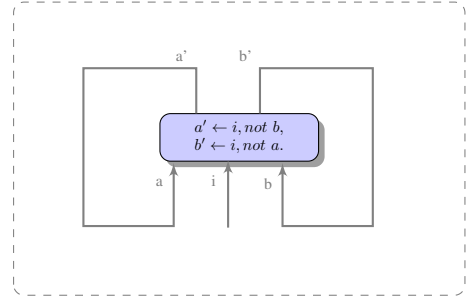
\begin{figure}
  \centering
\begin{tikzpicture}[scale=0.6,transform shape]

 \path \primitivemodule {M4}{$\begin{array}{l} 
a'\leftarrow i, not \ b, \\
b'\leftarrow i, not \ a. 
\end{array}$};

 \path [line] ($(modM4.north west)!0.23!(modM4.north east)$) -- +(-0,1) 
 node[left=9pt, above=1pt]{a'}-- +(-2.5,1) -- +(-2.5, -2.5) -- +(-0.5, -2.5) -- +(-0.5, -1.1) node[left=9pt, below=6pt]{a};
 
 \path [line] ($(modM4.north west)!0.63!(modM4.north east)$) -- +(-0,1) 
  node[left=9pt, above=1pt]{b'} -- +(2.5,1) -- +(2.5, -2.5) -- +(0.5, -2.5) -- +(0.5, -1.1) node[left=9pt, below=6pt]{b};
 
\path [line] ($(modM3.north west)!0.43!(modM3.north east)$) -- +(-0,1.5) node[left=9pt, below=6pt]{i};

 \begin{pgfonlayer}{background}
    \coordinate (A) at (-5, -3.5);
    \coordinate (B) at (5, 3);
    \path[rounded corners, draw=black!50, dashed]
   (A) rectangle (B);
 \end{pgfonlayer}

\end{tikzpicture}
  \caption{Module  $M_2$.}\label{fig:simple-modular-system}
\end{figure}

This modular system is deterministic, -- for each input (each of the eight possible interpretations of $i$, $a$ and $b$), there is at most one model. 
Notice also that 
$$
 \pi_{\{i,a,b\}} (M_1 [a=a'][b=b']) =M_0.
$$

\end{example}

\end{document}